\documentclass[11pt,a4paper]{article}
\usepackage{amsthm,fullpage,refcount}
\usepackage[dvipsnames]{xcolor}
\usepackage{tikz}
\usetikzlibrary {decorations.pathmorphing}
\usepackage{amsfonts}

\newcommand{\rank}{\ensuremath{\mathrm{rank}}}
\newcommand{\select}{\ensuremath{\mathrm{select}}}

\newtheorem{theorem}{Theorem}
\newtheorem{lemma}[theorem]{Lemma}
\newtheorem{corollary}[theorem]{Corollary}
\newtheorem{definition}{Definition}
\newtheorem{example}{Example}
\newcommand{\dd}{\mathinner{.\,.}}

\begin{document}

\title{Incongruity-sensitive access to highly compressed strings}

\author{Ferdinando Cicalese, Zsuzsanna Lipt\'ak, Travis Gagie,\\
Gonzalo Navarro, Nicola Prezza and Cristian Urbina}

\maketitle

\begin{abstract}
\noindent
Random access to highly compressed strings --- represented by straight-line programs or Lempel-Ziv parses, for example --- is a well-studied topic.  Random access to such strings in strongly sublogarithmic time is impossible in the worst case, but previous authors have shown how to support faster access to specific characters and their neighbourhoods.  In this paper we explore whether, since better compression can impede access, we can support faster access to relatively incompressible substrings of highly compressed strings.  We first show how, given a run-length compressed straight-line program (RLSLP) of size $g_{rl}$ or a block tree of size $L$, we can build an $O (g_{rl})$-space or an $O (L)$-space data structure, respectively, that supports access to any character in time logarithmic in the length of the longest repeated substring containing that character.  That is, the more incongruous a character is with respect to the characters around it in a certain sense, the faster we can support access to it.  We then prove a similar but more powerful and sophisticated result for parsings in which phrases' sources do not overlap much larger phrases, with the query time depending also on the number of phrases we must copy from their sources to obtain the queried character.
\end{abstract}

\section{Introduction}
\label{sec:introduction}

There has been a lot of work~\cite{BFHMP24,BCGGKNOPT21,BCPT15,BLRSSW15,CLLPPSS05,CU24,GNP20,GanardiContracting,KK20,KP18,KS22,KNP22,LMN24,Ryt03,shibata2025lzse,VY13} (see~\cite{KK25} for a recent discussion) on supporting fast worst-case random access to highly compressed strings.  For the purposes of this paper, ``highly compressed'' means stored in space that depends polynomially on a strong measure of compressibility --- the size of the smallest straight-line program, the substring complexity, or the size of the smallest string attractor, for example --- and only polylogarithmically on the length of the string.  We now know, for example, that we cannot support strongly worst-case sublogarithmic-time access to highly compressed strings~\cite{VY13}, but also how to store a block tree~\cite{BCGGKNOPT21,KNP22} for a string $S [1\dd n]$ over an alphabet of size $\sigma$ in optimal $O \left( \delta \log \frac{n \log \sigma}{\delta \log n} \right) \subseteq O (\delta \log n)$ space, where $\delta$ is the substring-complexity~\cite{KNP22}, and support access in worst-case $O (\log (n / \delta))$ time.  Other researchers have shown how to bookmark~\cite{GGKNP14,CGO16} a character or put a dynamic finger~\cite{BCCG18} on it, such that queries regarding neighbouring characters can be answered in time logarithmic in their distance from the bookmark or finger.

Worst-case bounds are only concerned with the characters hardest to reach, however, and bookmarks and fingers can be applied to arbitrary characters.  It is natural to ask whether the standard representations of highly compressed strings --- straight-line programs or LZ77 parses, for example --- treat all characters more or less equally or automatically make some characters easier to reach.  In fact, over twenty years ago Gasieniec et al.~\cite{GP03,GKPS05} showed how we can easily access characters at the ends of non-terminals' expansions in a straight-line program, and Claude and Navarro~\cite{CN12} used this in their grammar-based indexes.  More generally, though, if better compression tends to slow down random access, then perhaps it is easier to access character in less compressible areas of a compressed string.  Those areas are typically the most interesting ones when analyzing repetitive collections, for example they correspond to mutations or rare variants in genomes. In this paper we take a step toward demonstrating this by showing how some standard representations of highly compressible strings can be augmented --- without asymptotically increasing their space usage --- to support fast access to characters that are incongruous with respect to their neighbours, in the sense that the longest repeated substrings containing those characters are short.

In Section~\ref{sec:prelims_RLSLPs} we review notation and definitions we need for working with straight-line programs (SLPs) and run-length SLPs (RLSLPs).  In Section~\ref{sec:grammars} we first show how, given an RLSLP of size $g_{rl}$ for $S$ that is balanced in a certain sense, we can build an $O (g_{rl})$-space data structure that supports access to any character $S [q]$ in time $O (\log \ell_q)$, where $\ell_q$ is the length of the longest repeated substring containing $S [q]$.  It follows that we can store $S$ in $O(\delta\log\frac{n\log\sigma}{\delta\log n})$ space and support $O (\log \ell_q)$-time access.  We then extend our result to block trees and arbitrary RLSLPs.  In Section~\ref{sec:prelims_parses} we review notation and definitions we need for working with parses.  We then prove a similar but more powerful and sophisticated result for parsings: Given a parsing with $t$ phrases in which any phrase's source overlaps only phrases at most $w^{O (1)}$ times larger, where $w$ is the machine word size, we can store $S$ in $O (t)$ space and access $S [q]$ in $O (h_q + \log_w \ell_q)$ time, where the height $h_q$ of $S [q]$ is the number of phrases we must copy from their sources to obtain $S [q]$.  We conclude by showing that we can convert any bidirectional parse with $b$ phrases into one with this constraint on source-phrase overlaps, at the cost of an $O (\log_w (n / b))$-factor increase in size but with no character's height increasing.  It follows that, given a bidirectional parsing with $b$ phrases, we can build an $O (b \log_w (n / b))$-space data structure supporting $O (h_q + \log_w \ell_q)$-time access.

\section{Preliminaries for RLSLPs}
\label{sec:prelims_RLSLPs}

We denote by $[i,j]$ the integer interval $\{i, i+1, \dots, j\}$, and $[i,j) = [1,j-1]$.
Let $\Sigma=\{a_1,\dots,a_\sigma\}$ be a finite set of characters called the \emph{alphabet}. A \emph{string}  is a sequence $S[1\dd n]=S[1]\cdots S[n]$ of length $|S|=n$, where $S[i]\in \Sigma$ for all $i \in [1,n]$. The unique string of length 0, called the \emph{empty string}, is denoted $\epsilon$. A \emph{substring} of $S[1\dd n]$ is any subsequence $S[i\dd j]= S[i]\cdots S[j]$, where $i,j \in [1,n]$, and we let $S[i\dd j] = \varepsilon$ when $i > j$. If $i = 1$ then $S[i\dd j]$ is called a \emph{prefix}, and analogously, if $j = n$ then $S[i\dd j]$ is called a \emph{suffix}.
If $S[1\dd n]$ and $T[1\dd m]$ are strings, their concatenation is the string $S\cdot T[1\dd n+m]=S[1]\cdots S[n]T[1]\cdots T[m]$. We denote $S^k$ the concatenation of $k$ copies of the string $S$.  The \emph{substring complexity} of $S$, usually written $\delta$, is the minimum over $k$ of $1 / k$ times the number of distinct $k$-tuples in $S$.  To simplify notation, by $\log$ we will mean $\log_2$, and when writing $\log_dx$ we will actually mean $\max(1,\log_dx)$.  We use the transdichotomous word RAM model of computation with words of size $w\ge\log n$, where $n$ is the length of the uncompressed string.

A {\em straight-line program (SLP)} is a context-free grammar in Chomsky normal form (i.e., rules are of the form $A \to a$ or $A \to BC$, for nonterminals $A$, $B$, $C$ and terminal $a$) that generates exactly one string, $S[1\dd n]$. There is exactly one rule per nonterminal. We call $\exp(A)$ the string nonterminal $A$ expands to, with $\exp(a)=a$ for terminals $a$.

A generalization of SLPs called \emph{run-length SLP (RLSLP)} allows rules of the form $A \to B^k$ for integer $k > 2$, meaning $\exp(A) = \exp(B)^k$. The size of an SLP or RLSLP is defined as its number of rules and usually called $g$ or $g_{rl}$, respectively.

The {\em parse tree} of an SLP $G$ is a binary tree whose root is labeled with the start symbol of $G$, the leaves with terminal symbols (i.e., the symbols of $S$ if read left to right), and every internal node labeled with a nonterminal $A$, whose rule is $A \to BC$, has two children, the left one labeled $B$ and the right one labeled $C$. In case of RLSLPs, rules $A \to B^k$ are represented as a node labeled $A$ with $k$ children labeled $B$. In all cases, the parse tree has exactly $n$ leaves.

Any symbol $S[q]$ from an SLP or RLSLP can be extracted in time proportional to the height of its parse tree. We start aiming to extract $\exp(A)[q]$, where $A$ is the start symbol. In general, to extract $\exp(A)[q]$, where $A \to BC$, we recursively extract $\exp(B)[q]$ if $|\exp(B)| \ge q$, and $\exp(C)[q-|\exp(B)|]$ otherwise. For RLSLP rules $A \to B^k$, we reduce $\exp(A)[q]$ to $\exp(B)[1+((q-1) \bmod |\exp(B)|)]$. The recursion ends with $\exp(A)[1] = a$ when $A \to a$.

We will also make use of the so-called {\em grammar tree}. This is obtained by pruning from the parse tree, for each nonterminal label $A$, all but one occurrence of nodes with label $A$. Pruning means converting those other internal nodes to leaves. For RLSLP rules $A \to B^k$, the internal grammar tree node labeled $A$ has a left child labeled $B$ (which can be pruned or not) and a special right child, labeled $B^{k-1}$, which is a leaf. It is easy to see that the grammar tree of an SLP or RLSLP has exactly $2g+1$ or $2g_{rl}+1$ nodes, respectively~\cite{Nav20}. The grammar tree enables algorithms that run on the parse tree to use space $O(g)$ or $O(g_{rl})$ instead of $O(n)$.

An SLP or RLSLP is {\em balanced} if it has height $O (\log n)$ and {\em locally balanced}~\cite[Def~4.2]{CEKNP20} if is there is a constant $c$ such that every nonterminal $A$ has height at most $c \cdot \log |\exp(A)|$ in the parse tree.  A balanced SLP or RLSLP therefore supports worst-case logarithmic-time access.  With a locally balanced SLP or RSLP, moreover, when given any nonterminal $A$ and a position $q$ we can access $\exp (A) [q]$ in $O (\log |\exp (A)|)$ time by recursively descending from $A$ to a leaf in the parse tree.

\section{Incongruity-sensitive access with RLSLPs and block trees}
\label{sec:grammars}

Suppose we are given a locally balanced SLP or RLSLP.  The $g'$ leaves of its grammar tree naturally partition $S$ into $g'$ substrings: let $A_1,\ldots,A_{g'}$ be the leaves' labels, where $A_i$ can be a terminal or a nonterminal (or, with RLSLPs, an iteration $B^{k-1}$). The $i$th substring of the partition is then $\exp(A_i)$, so $S = \exp(A_1) \cdots \exp(A_{g'})$. Let $x_1,\ldots,x_{g'}$ be the starting positions of those substrings in $S$, and $x_{g'+1}=n+1$, that is, $x_1=1$ and $x_{i+1} = x_i + |\exp(A_i)|$. We build a predecessor data structure on the values $x_i$ for $1 \le i \le g'$, storing with each $x_i$ the label $A_i$. This yields the following algorithm to extract $S[q]$:
\begin{enumerate}
\item use a predecessor query to find the largest $x_i \le q$; 
\item extract the $(q-x_i+1)$th symbol of $\exp(A_i)$.
\end{enumerate}

With a standard predecessor data structure the time for the predecessor query could dominate the $O (\log |\exp (A_i)|)$ extraction time.  We can avoid this, however, by using a distance-sensitive predecessor data structure~\cite{BBV12,ehrhardt2017delta}.  From now on we focus on RLSLPs, which are more general.

\begin{lemma} \label{lem:grammar}
Let a locally balanced RLSLP of size $g_{rl}$ generate $S[1\dd n]$. Then there exists a data structure of size $O(g_{rl})$ that can access any $S[q]$ in time $O(\log \ell_q)$.
\end{lemma}

\begin{proof}
We define $x_1,\ldots,x_{g'}$ as above and implement the 2-point algorithm above. For the first point, we use a linear-space distance-sensitive predecessor data structure, which uses $O(g') = O(g_{rl})$ space and provides query time loglogarithmic on the distance between the query $q$ and its answer $x_i$, that is, $O(\log\log (q-x_i+1)) \subseteq O(\log\log(x_{i+1}-x_i))$; note $x_{i+1}-x_i=|\exp(A_i)|$.

We note that no leaf of the grammar tree labeled by a nonterminal can contain a unique substring, because by definition there must be another internal node with the same label, thus generating the same substring. This also holds for leaves labeled $B^{k-1}$, as there are at least two occurrences of $\exp(B)^{k-1}$ in $\exp(A) = \exp(B)^k$. It follows that $x_{i+1}-x_i \le \ell_q$ and thus the predecessor data structure takes time $O(\log\log\ell_q)$.

Point 2 takes time $O(\log (x_{i+1}-x_i)) \subseteq O(\log\ell_q)$ because the RLSLP is locally balanced: each recursive step reduces the height of the current parse tree node by at least 1.
\end{proof}

This is similar in flavor to how we can shorten AVL grammars to have height $O (\log (n / g))$ rather than $O (\log n)$~\cite{Ryt03}, and support $O (\log (n / g))$-time access on them, but that technique does not yield incongruity-sensitive access bounds.

Charikar et al.'s $\alpha$-balanced grammars \cite{CLLPPSS05} and Rytter's AVL grammars \cite{Ryt03} are locally balanced, for example, so they can be used to support the faster access time. So are the RLSLPs of Christiansen et al.~\cite{CEKNP20} and of Kociumaka et al.~\cite{KNO23}. The space bound of the second one \cite{KNO23} yields the following result.

\begin{corollary}
\label{cor:delta}
Let $S[1\dd n]$, over alphabet $[1\dd \sigma]$, have substring-complexity measure $\delta$. Then there exists a data structure of size $O(\delta\log\frac{n\log\sigma}{\delta\log n})$ that can access any $S[q]$ in time $O(\log \ell_q)$.
\end{corollary}

The same idea can be applied to other compressed hierarchical representations not based on grammars. In particular, a {\em block tree}~\cite{BCGGKNOPT21} on $S$ is defined by hierarchically partitioning it into perfect halves until the leaves. Every node is called a block and represents a substring of $S$.
We then prune every internal block of length $m$ such that the substrings of length $2m$ starting with it and ending with it occur earlier in $S$ (but not overlapping it). Those pruned nodes become then leaves in the block tree, and store a pointer to their source, which are the (one or two) blocks of the same length covering the leftmost occurrence of the block's content in $S$. A symbol $S[q]$ is accessed recursively, by going down left or right from the internal nodes, or shifting to the source of leaves. Since we always go down after having shifted to a source, and child block lengths are half the length of the block, this process takes time $O(\log n)$.

The block tree satisfies almost the same conditions we have used to prove Lemma~\ref{lem:grammar}, if we replace grammar-tree leaves by block-tree leaves. 

\begin{corollary} \label{cor:block-tree}
Let $S[1\dd n]$ be represented by a block tree of size $L$. Then there exists a data structure of size $O(L)$ that can access any $S[q]$ in time $O(\log \ell_q)$.
\end{corollary}
\begin{proof}
We define $x_1,\ldots,x_L$ as the starting position of the block tree leaves, and $x_{L+1}=n+1$. We build and use the same distance-sensitive predecessor data structure of the proof of Theorem~\ref{lem:grammar}, which uses $O(L)$ space and finds the leaf containing $S[q]$ in time $O(\log\log m)$, where $m=x_{i+1}-x_i$ is the leaf length.

The string of the leaf that contains $S[q]$ occurs earlier in $S$, since that leaf was created. It follows that $\ell_q \ge m$. The process of extracting $S[q]$ from the position $q-x_i+1$ of the leaf takes time $O(\log m) = O(\log \ell_q)$.
\end{proof}

By proper parameterization, the size of block trees can be bounded by $O(\delta\log\frac{n\log\sigma}{\delta\log n})$ \cite[Cor.~VI.7]{KNP22}, so we could prove Corollary~\ref{cor:delta} from Corollary~\ref{cor:block-tree}.  This could be useful because block trees can be smaller than RLSLPs on some texts.

We now aim at locally-balancing a grammar in order to extend Lemma~\ref{lem:grammar} to every RLSLP. There have been several useful results about balancing SLPs and RLSLPs recently, but none that exactly fit our goal.  Ganardi et al.~\cite{GJL21} showed how, given an SLP of size $g$, we can build a balanced SLP of size $O (g)$ of the same string.  Navarro et al.~\cite{NOU22} generalized this result to RLSLPs, but the resulting grammars might not be locally balanced:
even if the parse tree has height $O (\log n)$, if the nonterminal $A_i$ whose expansion we extract from in point (2) of our algorithm has height $\omega (\log |\exp (A)|)$, the extraction can take $\omega (\log \ell_q)$ time.

Ganardi~\cite{GanardiContracting} showed it is generally impossible, given an SLP of size $g$ for $S$, to build an SLP of size $O (g)$ for $S$ whose parse tree is weight-balanced (the sizes of any two siblings' subtrees are within constant factors of each other) or just height-balanced (the heights of any two siblings' subtrees differ by at most a constant).  Fortunately, just as all weight-balanced trees are height-balanced but not vice versa,
all SLPs with height-balanced parse trees are locally balanced but not all locally balanced SLPs have height-balanced parse trees. 
Ganardi therefore left open the possibility of there always being a locally balanced SLP of size $O (g)$ for $S$, and in fact he showed how we can build one.  In Appendix~\ref{app:contracting} we generalize Ganardi's construction from SLPs to RLSLPs and prove the following theorem, which may be of independent interest, in order to extend Lemma~\ref{lem:grammar} to arbitrary RLSLPs.

\begin{theorem}
\label{thm:local_balancing}
Given an RLSLP for $S$ with $g_{rl}$ rules, in $O (g_{rl})$ time we can build a locally balanced RLSLP for $S$ with $O (g_{rl})$ rules.
\end{theorem}

\begin{corollary} \label{cor:grammar}
Let an RLSLP of size $g_{rl}$ generate $S[1\dd n]$. Then there exists a data structure of size $O(g_{rl})$ that can access any $S[q]$ in time $O(\log \ell_q)$.
\end{corollary}

\section{Preliminaries for parses}
\label{sec:prelims_parses}

In general, a {\em parsing} cuts $S[1\dd n]$ into substrings called \emph{phrases}, $S = B_1 \cdots B_t$, each $B_i$ being either explicit (and of length at most $O(\log_\sigma n)$ and fitting into $O(1)$ space) or a copy of another substring $S[p_i\dd p_i+|B_i|-1]$ called $B_i$'s \emph{source}.  It follows that the phrase containing $S [q]$, if not explicit, has length at most $\ell_q$.  The number of phrases is called the parsing's \emph{size}.  If phrase's sources are always to their left or always to their right then the parsing is called \emph{unidirectional}, otherwise it is called \emph{bidirectional}.

Let $x_1,\ldots,x_t$ be the starting positions of the phrases $B_1,\ldots,B_t$, and $x_{t+1} = n+1$. To access $S$ from the parsing, we define $f(q)=q$ if the phrase $B_i$ that $S[q]$ belongs to is explicit, and $f(q) = p_i + q - x_i$ otherwise; note that $S[q] = S[f(q)]$. To obtain $S[q]$, we apply $f$ iteratively until we reach a character in an explicit phrase.  If this never happens then $S$ cannot be decoded from the parse.
 
The \emph{height} $h_q$ of $S [q]$ is the number of times we must apply apply $f$ to $q$ before reaching a character in an explicit block, and $S [q]$'s {\em referencing chain} is $q, f(q), \ldots, f^{h_q}(q)$. We can thus extract $S[q]$ in $h_q +1$ steps. If we use a predecessor data structure to find the phrase containing $f^k(q)$ for $0 \leq k \leq h_q$ --- that is, each $B_i$ such that $x_i$ is the maximum value with $x_i \le f^k(q)$ --- then we can access $S[q]$ in time $O((h_q +1) \log\log n)$.  We write simply $h$ to denote the maximum height of any character in $S$.

\section{Incongruity-sensitive access with parses}
\label{sec:parses}

We have actually already parsed $S$ for Lemma~\ref{lem:grammar}, into $S = \exp(A_1) \cdots \exp(A_{g'})$, since each leaf in the grammar tree is either a terminal or a non-terminal that appears as an internal node elsewhere in the grammar tree.  Indeed, we can view Lemma~\ref{lem:grammar} as applying to a special kind of parse with two useful properties: first, each character $S [q]$ has height $h_q$ at most logarithmic in the length of the phrase containing $S [q]$; second, sources' boundaries align nicely with phrases' boundaries in such a way that when extracting $S [q]$ we need perform only one predecessor query, to find the phrase containing $S [q]$.  Needless to say, most parsings do not have these properties.

Typical parsings~\cite{LZ76,SS82} have good bounds on $t$, but do not provide upper bounds for the maximum height $h$ other than the number $t$ of phrases and sometimes only the length $n$ of the uncompressed string.  For example, the original Lempel-Ziv parse \cite{LZ76} (whose parsing size is called $z$) and bidirectional macro schemes \cite{SS82} (whose optimal parsing size is called $b$) can achieve less space than the smallest grammars and still recover $S$. It always holds $\delta \le b \le z \le 2g_{rl}^* \le 2g^*$, where $g^*$ and $g_{rl}^*$ are the sizes of the smallest SLP and RLSLP for $S$, respectively; the differences are asymptotically strict for some string families \cite{Nav20.1}. 

Some parsings increase $t$ in order to obtain better upper bounds on $h$~\cite{KP18,KS22,LMN24,BFHMP24,shibata2025lzse}, but those with $t \in O (\gamma \log (n / \gamma))$ do not achieve $h = o (\log (n/t))$.  Here $\gamma$, with $\delta \le \gamma \le b$, is the size of the smallest string attractor for $S$ \cite{KP18}, so the minimum number of characters in $S$ we need to select such that at least one selected character appears in some occurrence of any substring of $S$.  For example, we can have $h = O(\log n)$ with $t = O(g_{rl}^*)$ with a unidirectional parsing~\cite{BFHMP24}, or $h=O(\log(n/\gamma))$ with $t=O(\gamma\log(n/\gamma))$ with a bidirectional parsing~\cite{KP18}. For some string families, these parsings can have $t$ asymptotically smaller than the sizes of the smallest RLSLPs or block trees.

Our main result in this section is to improve the access time for $S [q]$ from our baseline of $O ((h_q + 1) \log \log n)$ to $O (h_q + \log_w \ell_q)$ for a class of ``well-behaved'' parses. 

\begin{definition}
A parse is {\em $\alpha$-contracting} if no source overlaps a phrase that is strictly more than $\alpha$ times longer than it.
\end{definition}

Notice that a $c$-contracting parsing is also $c'$-contracting for any $c' \geq c$. Crucially, if a parse is $\alpha$-contracting with $\alpha \leq 1 - \epsilon$ for some positive constant $\epsilon$, then $h \in O (\log n)$ and, moreover, $h_q \in O (\log \ell_q)$ for all $q$.

While the parsings induced by RLSLPs and block trees are naturally $1$-contracting and $1 / 2$-contracting, respectively, general parses are not even $o (n)$-contracting.  We will be able, however, to provide efficient access to other $\alpha$-contracting parses with $\alpha$ polynomial in the word size $w \geq \log n$, as long as the accessed characters have small height. This will turn out to be significantly more challenging.

\begin{theorem} \label{thm:access on alpha-balanced}
Let $S[1\dd n]$ have an $\alpha$-contracting parse of $t$ phrases for any $\alpha \le w^{O(1)}$. Then there exists a data structure of size $O(t)$ that can access any $S[q]$ in time $O(h_q + \log_w \ell_q)$.
\end{theorem}

This result encompasses those of Section~\ref{sec:grammars}, as those are all $1$-contracting parsings with $h_q = O(\log\ell_q)$.  It also applies to other $\alpha$-contracting parsings not based on grammars. For example, Kempa and Prezza \cite[Thm 3.12]{KP18} describe a bidirectional parse of size $t = O(\gamma\log(n/\gamma))$ where the sources only overlap phrases no longer than half the phrase's size, implying $\alpha=1/2$ and $h_q = O(\log \ell_q)$, so they obtain $O(\log\ell_q)$ access time.

We begin the proof of Theorem~\ref{thm:access on alpha-balanced} with a modified construction (based on a classic result due to Gilbert and Moore~\cite{GM59}) of Bille et al.'s~\cite{BLRSSW15} interval-biased search trees:

\begin{lemma}
\label{lem:subintervals2}
Suppose we are given a partition of the interval $[1, n)$ into $s$ subintervals
\[\left[ \rule{0ex}{2.5ex} x_1 = 1, x_2 \right), [x_2, x_3), \ldots, [x_{s - 1}, x_s), \left[ \rule{0ex}{2.5ex} x_s, x_{s + 1} = n \right)\,.\]
For any integer $d\ge 2$, we can build a $d$-ary tree $T$ storing keys $x_1, \ldots, x_s$ (regarded as sequences of $d$-ary digits) at its leaves, in which each $x_i$ has depth
\[\left\lfloor \log_d \frac{n}{x_{i + 1} - x_i} \right\rfloor + 2\,.\]
\end{lemma}

\begin{proof}
For $1 \leq i \leq s$, let $\beta_i$ be the string consisting of the first
\[\left\lfloor \log_d \frac{n}{x_{i + 1} - x_i} \right\rfloor + 2\]
digits to the right of the base-$d$ point in the base-$d$ representation of $\frac{x_i + x_{i + 1}}{d n}$.

Assume $j>i$ (the other case is symmetric). Rewrite $(x_j+x_{j+1})-(x_i+x_{i+1}) = x_{j+1} - (x_i+x_{i+1}-x_j)$. Since $x_{j+1} > x_j \ge x_{i+1}$ we have, for the last term, $x_i+x_{i+1}-x_j \le x_i$, and thus the right-hand of the inequality is $x_{j+1} - (x_i+x_{i+1}-x_j) \ge x_{j+1}-x_i > x_{i+1}-x_i$.
This means that $(x_j+x_{j+1})$ differs from $(x_i+x_{i+1})$ by at least $x_{i + 1} - x_i$ for all $j\neq i$. Hence, $\frac{x_i + x_{i + 1}}{d n}$ differs from $\frac{x_j + x_{j + 1}}{d n}$ by at least $\frac{x_{i + 1} - x_i}{d n}$ so their base-$d$ representations agree on fewer than  $\log_d \frac{dn}{x_{i + 1} - x_i} \le \log_d \frac{n}{x_{i + 1} - x_i} + 1$ digits to the right of their base-$d$ points.  It follows that $\beta_i$ is not a prefix of $\beta_j$.

We build the $d$-ary trie for $\beta_1, \ldots, \beta_s$, ordering nodes' children by the natural order induced by digits $[0,d-1]$, store $x_1, \ldots, x_s$ at the leaves, and discard the edge labels.
\end{proof}

Next, we prove a lemma about the trie $T$ used in our proof of Lemma~\ref{lem:subintervals2}, similar to Gagie et al.'s~\cite[Sec.~2]{GGKNP12} observation about the structure of parse trees of locally balanced SLPs:

\begin{lemma}
\label{lem:leaves2}
Given $l $ and $r$ with $1 \leq l  < r \leq s$, there exist $d' \le d$ nodes $v_{1}, \dots, v_{d'}$ at depths at least $\log_d \frac{n}{x_r + x_{r + 1} - x_l  - x_{l  + 1}}$ in $T$ whose subtrees together contain the $l $th through $r$th leaves. 
Furthermore, the lowest common ancestor $u$ of the $l $th and $r$th leaves is such that (i) $lca(v_a,v_b)=u$ for all $1 \le a < b \le d'$, and (ii) the path from $u$ to $v_a$ (both excluded) contains only nodes with out-degree equal to one for all $1 < a < d'$.
\end{lemma}

\begin{proof}
Recall that $u$ is defined to be the lowest common ancestor of the $l $th and $r$th leaves. 
Let $u^l  \neq u^r$ be the children of $u$ being ancestors of the $l $th and $r$th leaves, respectively. The integer $d'\le d$ will be the number of siblings (children of $u$) between $u^l $ and $u^r$, both included.  This is illustrated in Figure~\ref{fig:trie_T} in Appendix~\ref{app:fig_for_leaves2}.

\paragraph{Finding $v_1$}
We choose $v_1$ to be the lowest node on the path from $u^l $ to the $l $th leaf such that $v_1$ is an ancestor of the rightmost leaf in the subtree rooted in $u^l $.

Letting $v$ be a trie node and $c\in [0,d-1]$ be a base-$d$ digit, let us denote with $v(c)$ the child of $v$ by digit $c$.
For the sake of a contradiction, assume $v_1$'s depth is less than $\log_d \frac{n}{x_r + x_{r + 1} - x_l  - x_{l  + 1}}$. Consider the path label $\beta$ of $v_1$.
Let $y$ be the base-$d$ digit corresponding to the child $v_1(y)$ of $v_1$ being an ancestor of the rightmost leaf in the subtree rooted in $u^l $.
The numbers in $[0, 1)$ whose base-$d$ representations start with $0.\beta y$ form a range $R$ of size $|R| = d^{- (|\beta| + 1)}$.
Since above we assume $|\beta| < \log_d \frac{n}{x_r + x_{r + 1} - x_l  - x_{l  + 1}}$, the size of $R$ satisfies $|R| = d^{- (|\beta| + 1)} > \frac{x_r + x_{r + 1} - x_l  - x_{l  + 1}}{d n}$.

First, observe that $R$ must end after $\frac{x_l +x_{l +1}}{dn}$: this follows from the definition of $v_1$, which implies that the path label of the $l $th leaf starts with $\beta y'$ for some $y'<y$.
Additionally, since $R$ cannot include $\frac{x_r + x_{r + 1}}{d n}$ --- otherwise $v_1$ would be an ancestor of the $r$th leaf, contrary to our choices of $u$ and $v_1$ --- it must start before 

\[\frac{x_r + x_{r + 1}}{d n} - \frac{x_r + x_{r + 1} - x_l  - x_{l  + 1}}{d n}
= \frac{x_l  + x_{l  + 1}}{d n}\,.\]
This means that $R$ includes $\frac{x_l  + x_{l  + 1}}{d n}$, which implies that the base-$d$ representation of $\frac{x_l  + x_{l  + 1}}{d n}$ starts with $0.\beta y$, so both the $l $th leaf and the rightmost leaf in the subtree rooted in $u^l $ are in the subtree rooted in $v_1(y)$, contrary to our choice of $v_1$: a contradiction.  Therefore, we conclude that $v_1$'s depth is at least $\log_d \frac{n}{x_r + x_{r + 1} - x_l  - x_{l  + 1}}$.

\paragraph{Finding $v_{d'}$}
The choice of $v_{d'}$ is symmetric to that of $v_1$. 
Recall that $u^r$ is the child of $u$ being an ancestor of the $r$th leaf.
We choose $v_{d'}$ to be the lowest node on the path from $u^r$ to the $r$th leaf such that $v_{d'}$ is an ancestor of the leftmost leaf in the subtree rooted in $u^r$.
A symmetric argument shows that $v_{d'}$'s depth is also at least $\log_d \frac{n}{x_r + x_{r + 1} - x_l  - x_{l  + 1}}$.

\paragraph{Finding $v_2, \dots, v_{d'-1}$} If $d'>2$, we are left to show how to find nodes $v_2, \dots, v_{d'-1}$. Let $w_k$, for $k=1, \dots, d'$, denote the $k$th child of $u$ starting from $u^l $ included (that is, $u^l  = w_1$, $w_2$ is the next sibling of $w_1$, and so on until $w_{d'}=u^r$). We define $v_{k}$ to be the deepest descendant of $w_k$ being an ancestor of all the leaves below $w_k$. In other words, $v_{k}$ is the shallowest descendant of $w_k$ with at least two children, or the unique leaf below $w_k$ if no such descendant exists (note $v_{k} = w_k$ if and only if $w_k$ has at least two children). By the definitions of $u$ and $v_1, \dots, v_{d'}$, it is clear that the subtrees below $v_1, \dots, v_{d'}$ contain the $l$th through $r$th leaves. We are left to show that also nodes $v_2, \dots, v_{d'-1}$ are at depths at least $\log_d \frac{n}{x_r + x_{r + 1} - x_l - x_{l + 1}}$. Consider any such node $v_q$, with $q \in [2,d-1]$. If $v_q$ is a leaf, say, the $t$th leaf, then its depth is at least $\log_d \frac{n}{x_{t + 1} - x_t}$, where $x_t \ge x_{l+1}$ and $x_{t+1} \le x_r$. Then, the depth of $v_q$ is at least $\log_d \frac{n}{x_{t + 1} - x_t} \ge \log_d \frac{n}{x_{r} - x_{l+1}} \ge \log_d \frac{n}{x_{r} + x_{r+1} - x_l - x_{l+1}}$ (the latter inequality comes from $x_{r+1} - x_l > 0$). 

If $v_q$ is not a leaf, consider any two leaves below $v_q$ --- let us call them the $l'$th and $r'$th leaves (with $r'>l'$ without loss of generality). Since it must be $l < l' < r' < r$ and, hence, $x_l < x_{l+1} \le x_{l'} < x_{l'+1} \le x_{r'} < x_{r'+1} \le x_r < x_{r+1}$, we have $\frac{x_{r'}+x_{r'+1}}{dn} - \frac{x_{l'}+x_{l'+1}}{dn} \le \frac{x_r + x_{r+1}-x_l - x_{l+1}}{dn}$, meaning that the lowest common ancestor of the $l'$th and $r'$th leaves is at depth at least $\log_d\frac{dn}{x_r + x_{r+1}-x_l - x_{l+1}} > \log_d\frac{n}{x_r + x_{r+1}-x_l - x_{l+1}}$. Since this holds for all pairs of leaves below $v_q$ and $v_q$ is the deepest node being an ancestor of all those leaves, we obtain that $v_q$ itself is at depth at least $\log_d\frac{n}{x_r + x_{r+1}-x_l - x_{l+1}}$.
\end{proof}

The next lemma follows from Lemmas~\ref{lem:subintervals2} and~\ref{lem:leaves2}; we give the proof in Appendix~\ref{app:main2}.

\begin{lemma} \label{lem:main2}
Suppose we are given the tree $T$ of Lemma~\ref{lem:subintervals2} for the partition 
\[\left[ \rule{0ex}{2.5ex} x_1 = 1, x_2 \right), [x_2, x_3), \ldots, [x_{s - 1}, x_s), \left[ \rule{0ex}{2.5ex} x_s, x_{s + 1} = n \right)\,,\]
and $t$ arbitrary subintervals $[y_1, z_1) \ldots, [y_t, z_t)$ of $[1, n)$, with $n, x_1, \ldots, x_s, y_1, \ldots, y_t, z_1, \ldots, z_t$ all integers.  
Then, we can build an $O (t + s)$-space data structure on top of $T$
with which, given $i$ and an integer $q$ in $[y_i, z_i)$, we can find $j$ such that $q$ is in $[x_j, x_{j + 1})$ by traversing at most $\max\left(0,\log_d \frac{z_i - y_i}{x_{j + 1} - x_j}\right) + O(1)$ edges of $T$ and performing $O(\log_w d)$ additional constant-time operations, where $w$ is the memory word size (in bits). 
\end{lemma}
\
The last ingredient in our proof of Theorem~\ref{thm:access on alpha-balanced} is a the following result, proven in Appendix~\ref{app:distance-sensitive}, about linear-space distance-sensitive predecessor data structures.  It improves the query time of Belazzougui et al.'s~\cite{BBV12} and Ehrhardt's~\cite{ehrhardt2017delta} from $O (\log \log \Delta)$ to $O (\log \log_w\Delta)$.

\begin{theorem}\label{thm:distance-sensitive}
    Let $X = \{x_1, \dots, x_s\} \subseteq [1,n]$. 
    We can build a data structure on $S$ taking $O(s)$ words of space and answering predecessor queries on $X$ in $O(\log\log_w\Delta)$ time, where $w\ge \log n$ is the machine word size in bits. 
\end{theorem}

Now we are finally in position to prove Theorem~\ref{thm:access on alpha-balanced}.

\begin{proof}[Proof of Theorem~\ref{thm:access on alpha-balanced}]
Let $x_1,\ldots,x_t$ be the starting positions in $S$ of the phrases, with $x_{t+1}=n+1$. Further, let $[y_i,z_i)$ be the source of the phrase $[x_i,x_{i+1})$ (if the phrase is explicit, this interval is omitted). We build the tree $T$ of Lemma~\ref{lem:subintervals2} on the values $x_i$, and the additional data structure of Lemma~\ref{lem:main2} on top of $T$, with the intervals $[y_i,z_i)$. Further, we build a distance-sensitive predecessor data structure on the values $x_1,\ldots,x_t$ according to Theorem~\ref{thm:distance-sensitive}.

The procedure to access $S[q]$ is then as follows:
\begin{enumerate}
    \item Find the predecessor $x_i$ of $q$ with the data structure of Theorem~\ref{thm:distance-sensitive}.
    \item If the phrase $[x_i,x_{i+1})$ is the explicit string $B_i$, return $B_i[q-x_i+1]$.
    \item Update $q \gets q - x_i + y_i$, which belongs to $[y_i,z_i)$.
    \item Given $i$ and $q$, use Lemma~\ref{lem:main2} to find $j$ such that $q$ is in $[x_j,x_{j+1})$.
    \item Update $i \gets j$ and return to point 2.
\end{enumerate}

The number of iterations of points 2--5 is at most $h_q$.
By Lemma~\ref{lem:main2}, the number of edges of $T$ traversed in point 4 is at most $\max\left(0,\log_d \frac{z_i-y_i}{x_{j+1}-x_j}\right) + O(1) = \max\left(0,\log_d \frac{x_{i+1}-x_i}{x_{j+1}-x_j}\right) + O(1)$. 
Furthermore, Lemma~\ref{lem:main2} spends an additional time $O(\log_wd)$ on top of that.
We now show that the sum of traversed edges telescopes along the iterations. 

Let $x_{i_0} = x_i$ be the value of $x_i$ obtained in point 1, and given $x_{i_k} = x_i$ in point 4, let $x_{i_{k+1}} = x_j$. Then the number of edges traversed in $T$ along the $h_q$ steps of the iteration is at most
$$ O(h_q) + \sum_{k=0}^{h_q-1} \max\left(0,\log_d\frac{x_{i_k+1}-x_{i_k}}{x_{i_{k+1}+1}-x_{i_{k+1}}}\right).$$
The sum would telescope, except that negative logarithms (i.e., when the source is longer than the phrase) introduce zeros in the sequence. However, since the parsing is $\alpha$-contracting, it follows that $\log_d \frac{x_{i_k+1}-x_{i_k}}{x_{i_{k+1}+1}-x_{i_{k+1}}} \ge -\log_d \alpha$. We can then upper bound all terms $\max\left(0,\log_d\frac{x_{i_k+1}-x_{i_k}}{x_{i_{k+1}+1}-x_{i_{k+1}}}\right)$ $\le \log_d \alpha + \log_d\frac{x_{i_k+1}-x_{i_k}}{x_{i_{k+1}+1}-x_{i_{k+1}}}$, and upper bound the summation  by
$$ O(h_q) + h_q \log_d \alpha + \sum_{k=0}^{h_q-1} \log_d \frac{x_{i_k+1}-x_{i_k}}{x_{i_{k+1}+1}-x_{i_{k+1}}},$$
which now does telescope to $\log_d (x_{i_0+1}-x_{i_0}) + O(h_q(1+\log_d\alpha))$.

Let $x_s = x_{i_0}$ and $l_q = x_{s+1}-x_s$ be the length of the block where $q$ lies.
We have just shown that the sum adds up to $\log_d(x_{s+1}-x_s) + O(h_q(1+\log_d\alpha)) = O(h_q(1+\log_d\alpha)+\log_d l_q)$. Each of those edge traversals can be done in constant time using perfect hashing to store the children of the nodes.
The cost of point 1 is $O(\log\log_w (q - x_s)) \subseteq O(\log\log_w l_q)$  
Taking into account the additional time $O(\log_wd)$ of each application of Lemma~\ref{lem:main2}, the total cost becomes $O(h_q(1+\log_d\alpha+\log_wd) + \log_d l_q + \log\log_w l_q)$. Choosing the arity of $T$ as $d=w$ and remembering that we assume $\alpha \le w^{O(1)}$, this cost becomes $O(h_q + \log_w l_q)$.

Because of the parsing invariants, a non-explicit phrase must occur elsewhere in $S$; therefore we can bound $l_q \le \ell_q$. This proves the claimed time complexity.

Both the predecessor data structure and the extra structures of Lemma~\ref{lem:main2} use $O(t)$ space. The tree $T$, however, has $t$ leaves but could have more internal nodes. To achieve the space bound of $O(t)$, we compact $T$ to obtain $T'$, replacing by a single edge every maximal path in $T$ with nodes having exactly 1 child each.
(Note $T'$ is still a $d$-ary search tree storing on each of its edges the smallest integer $x_i$ below it to guide search.)
This process may eliminate the nodes $v_1$ or $v_{d'}$ of Lemma~\ref{lem:leaves2} but, if $r'$ and $r''$ are their lowest ancestors in $T$ with 2 or more children each, then we can still start from $r'$ and $r''$ in $T'$ and find the $j$th leaf of $T'$ traversing the same amount of edges (or less). Note that, in $T'$, we will not have the problem of handling unary paths we have discussed above.
\end{proof}

In future work we will explore how Theorem~\ref{thm:access on alpha-balanced} can be applied.  For example, in Appendix~\ref{app:alpha bidirectional} we prove the following lemma about converting bidirectional parses into somewhat larger contracting bidirectional parses.

\begin{lemma}
\label{lem:alpha bidirectional}
Let $S[1\dd n]$ have a bidirectional parse of size $b$, and denote with $h_q$ the height of $S[q]$ in such a parse. For any integer $\alpha > 1$, we can build an $\alpha$-contracting bidirectional parse of size $O(b \log_{\alpha}(n/b))$ for $S$ such that the height $h'_q$ of any character $S[q]$ in such a parse is no larger than that in the original parse: $h'_q \le h_q$.
\end{lemma}

Setting $\alpha = w$ in Lemma~\ref{lem:alpha bidirectional}, where $w$ is the word size, and combining with Theorem~\ref{thm:access on alpha-balanced}, we obtain the following result.

\begin{corollary} \label{cor:optimal access}
Let $S[1\dd n]$ have a bidirectional parse of size $b$. On a machine word of word size $w \ge \log n$ bits, we can store a data structure of $O(b \log_w(n/b))$ words supporting the extraction of any $S[q]$ in time $O(h_q + \log_w \ell_q)$, where $\ell_q$ is the length of the longest repeated substring of $S$ containing position $q$ and $h_q$ the height of $S[q]$ in the parse. 
\end{corollary}

\appendix

\section{Proof of Theorem~\ref{thm:local_balancing}}
\label{app:contracting}

We show how to transform an RLSLP into a context-free grammar (with run-length rules, which we call a ``run-length grammar'') that generates the same string and is {\em contracting}, that is, if $B$ is a child of $A$, then $|\exp(B)|\le|\exp(A)|/2$. This property implies local balance. That contracting run-length grammar is not yet an RLSLP because some right-hand sides have length more than two, but still bounded by a constant. We finish by transforming the contracting run-length grammar again into an RLSLP by putting it into Chomsky Normal Form. This loses the contracting property but retains local balance (with an increased constant). 

We borrow some definitions and results from Ganardi's work on constructing contracting grammarrs \cite{GanardiContracting}, and show that (almost) the same ideas work to produce contracting run-length grammars.

Consider a (run-length) grammar $G = (V,\Sigma,R,S)$. We call the {\em variables}, $V \cup \Sigma$, the nonterminals and terminals of $G$. A variable $B \in V$ is a {\em heavy child} of $A \in V$ with rule $A \rightarrow u$ if $B$ appears in $u$ and $|\exp(B)|>|\exp(A)|/2$. 
A rule $A \rightarrow u$ is \emph{contracting} if $u$ contains no heavy variables. Naturally, run-length variables and terminal symbols are considered contracting. If all rules in $G$ are contracting, we call $G$ contracting.

A {\em labeled tree} $T = (V,E,\eta)$ is a rooted tree where each edge $e \in E$ is labeled by a string $\eta(e)$. 
A prefix in $T$ is the concatenation of the labels in any path starting from the root.

The {\em DAG} of $G = (V,\Sigma,R,S)$ is obtained from its grammar tree, by identifying every leaf labeled with a nonterminal $A$ with the only internal node labeled $A$, and all the leaves labeled with a terminal symbol with a single node. The DAG contains exactly one node labeled with each variable; we will identify nodes and labels.

The {\em heavy forest} $H = (V,E_H)$ of $G$ contains all edges $(A,B)$ where
$B$ is a heavy child of $A$, and it is a subgraph of the DAG of $G$. Notice that the edges in $H$ point towards the roots, that is, if $(A,B) \in E_H$ then $A$ is a child of $B$ in $H$. The ending points of those paths in $E_H$ (i.e., their lowest nodes in the DAG) are called {\em roots}.
Isolated nodes, in particular all terminal and run-length variables of $G$, are also roots.

Let's now assume $G$ is an (RL)SLP. We define two labeling functions: The {\em left label} $\lambda(e)$ of an edge $e = (A,B) \in E_H$ is the left light child of $A$ (if it exists, $\varepsilon$ otherwise) and the {\em right label} $\rho(e)$ of $e$ is the right light child of $A$ (if it exists, $\varepsilon$ otherwise). The connected components of $(H,\lambda)$ and $(H,\rho)$ are called the {\em left labeled and right labeled heavy trees}, which can be computed in linear time from $G$. If $B$ is the root of a heavy tree containing a variable $A$ we can factorize $A$ into the reversed left labeling from $A$ to its root $B$ in $H$,  $B$, and the right labeling of the path from $B$ to $A$. In that way one can redefine every variable using SLPs that define all prefixes in the left labeled and the right labeled heavy trees.

\begin{example}Let $S=abracad(abra)^7cabra$. Let $A,B,C,D,R$ be the variables generating $a,b,c,d$, and $r$, respectively. Below, we exhibit an  RLSLP that generates $S$.
\begin{center}
\begin{tabular}{lclclcl}
Rule & \hspace{1cm} & $\exp(\cdot)$ & \hspace{1cm} & $|\exp(\cdot)|$ &\hspace{1cm} & Heavy child \\
& & & & \\
    $A_0 \rightarrow A_1A_4$ & & $abracad(abra)^7cabra$ & & $40$ & & $A_1$ \\
    $A_1 \rightarrow A_2A_3$ & & $abracad(abra)^7$ & & $35$ & & $A_3$\\
    $A_2 \rightarrow A_5A_6$ & & $abracad$ & & $7$ & & $A_5$\\
    $A_3 \rightarrow A_5^7$  & & $(abra)^7$ & & $28$ & & None\\
    $A_4 \rightarrow CA_5$   & & $cabra$ & & $5$ & & $A_5$\\
    $A_5 \rightarrow A_7A_8$   & & $abra$ & & $4$ & & None\\
    $A_6 \rightarrow A_9D$   & & $cad$ & & $3$ & & $A_9$\\
    $A_7 \rightarrow AB$   & & $ab$ & & $2$ & & None\\
    $A_8 \rightarrow RA$   & & $ra$ & & $2$ & & None\\
    $A_9 \rightarrow CA$   & & $ca$ & & $2$ & & None\\
\end{tabular}
\end{center}
The heavy forest of this RLSLP is shown in Figure~\ref{fig:contracting}. Observe how the labels in the edges have $O(1)$ length, and the number of edges is no more than the number of variables of the RLSLP. 
\end{example}

\begin{figure}
\center
\begin{tikzpicture}
    \node[circle, draw] (A0) at (0,0) {$A_0$};
    \node[circle, draw] (A1) at (-1,-2) {$A_1$};
    \node[circle, draw] (A2) at (2,-2) {$A_2$};
    \node[circle, thick, draw] (A3) at (0,-4) {$A_3$};
    \node[circle, draw] (A4) at (4,-2) {$A_4$};
    \node[circle, thick, draw] (A5) at (3,-4) {$A_5$};
    \node[circle, draw] (A6) at (7,-2) {$A_6$};
    \node[circle, thick, draw] (A9) at (6,-4) {$A_9$};
    \draw[->] (A0) -- node[left]{[$\varepsilon,A_4$]}  (A1);
    \draw[->] (A1) -- node[left]{[$A_2,\varepsilon$]} (A3);
    \draw[->] (A2) -- node[left]{[$\varepsilon,A_6$]} (A5);
    \draw[->] (A4) -- node[right]{[$C,\varepsilon$]}(A5);
    \draw[->] (A6) -- node[right]{[$\varepsilon,D$]}(A9);
\end{tikzpicture}\caption{Heavy forest of an RLSLP generating the string $S=abracad(abra)^7cabra$ (we omit trees containing only one node). The edges are labeled with the left and right labels corresponding to each non-root variable. We use bold circles to highlight root nodes. }\label{fig:contracting}
\end{figure}
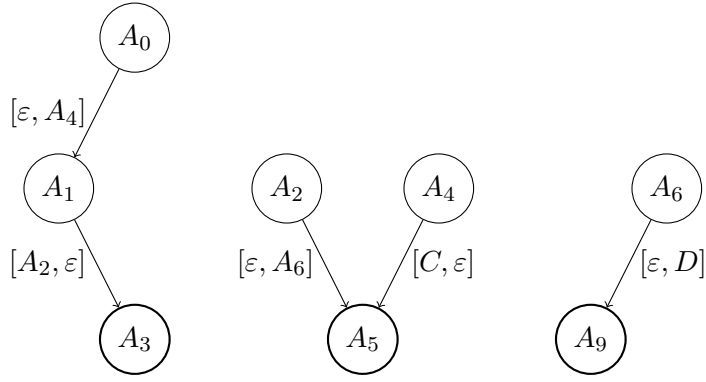

\begin{theorem}[{\cite[Thm. 6]{GanardiContracting}}]\label{thm:H_SLP_contracting} Given a labeled tree $T$ with $|T|$ edges and labels of length $\le t$, one can compute
 in linear time a contracting grammar with $O(|T|)$ variables and right-hand sides of length $O(t)$
 defining all prefixes in $T$.
\end{theorem}

The grammars obtained using Ganardi's procedure are not necessarily in Chomsky Normal Form. Their procedure actually allows that rules in its input grammars are of different lengths. This allows accommodating our run-length rules. Run-length variables, if unfolded, would be roots in $H$ under Ganardi's original definition of contracting variables. Hence, the labeled trees $H_L=(H,\lambda)$ and $H_R=(H,\rho)$ have $O(|G|)$ edges with labels of length at most 1. We obtain the following corollary.

\begin{corollary}\label{cor:contracting_heavy_tree}Let $G$ be an RLSLP. One can construct in $O(|G|)$ time a contracting run-length grammar $G_L$  of size $O(|G|)$ and right-hand sides of length $O(1)$ defining all prefixes in $H_L$.\end{corollary}

\begin{proof}
We apply Theorem \ref{thm:H_SLP_contracting} to all connected components of $H_L$.  As the number of edges in $H_L$ is no more than $|G|$, and the length of the labels  is at most $1$, the claim holds. 
\end{proof}

An analogous result holds for $H_R$. Now we are ready to prove Theorem \ref{thm:local_balancing}.

\begin{proof}[Proof of Theorem \ref{thm:local_balancing}]
Let $G$ be an RLSLP generating a string $S$ with size $g_{rl}$. We apply the construction of Corollary \ref{cor:contracting_heavy_tree}. For each non-root variable $A$ of a connected component of $H$, let $B$ be its root. Let $X_L\in G_L$ and $X_R\in G_R$ be variables generating the reversed left-labeled prefix and right-labeled prefix of $A$. We observe that $A$ derives the same string as $X_L\cdot B\cdot X_R$. To ensure $A$ is locally contracting, we replace the right-hand side of the rule of $A$ by the concatenation of the right-hand sides of the rules of $X_L, B,$ and $X_R$. Because $X_L,B,$ and $X_R$ are contracting, the resulting rule for $A$ is also contracting. Observe that the length of the new rule for $A$ is $O(1)$. The root variables (which includes run-length variables of $G$) are not modified. The size of this equivalent run-length grammar $G'$ is $O(g_{rl})$, and all its rules are contracting, hence $G'$ is contracting. Because rules are constant size, we can apply  standard techniques to transform the non-run-length rules of $G'$ into Chomsky Normal Form. This incurs only a constant increase in height, and so retains local balance.
\end{proof}

\section{Proof of Theorem~\ref{thm:distance-sensitive}}
\label{app:distance-sensitive}

\begin{definition}
    Let $X = \{x_1, \dots, x_s\} \subseteq [1,n]$ and let $x_0=0$. 
    For any $q\in [1,n]$, the \emph{predecessor} $q^-$ of $q$ in $X \cup\{0\}$ is defined as $q^- = \max\{x_i\ :\ i\in[0,s]\ \wedge\ x_i\le q\}$.
    The \emph{strict predecessor} of $q$ in $X\cup\{0\}$ is $\max\{x_i\ :\ i\in[0,s]\ \wedge\ x_i< q\}$.
    We denote as $\Delta = q-q^-$ the distance to $q$ from its predecessor.
    A predecessor query on $X$ finds such $q^-\in X\cup \{0\}$ given $q$ as input.
\end{definition}

Intuitively, to prove Theorem \ref{thm:distance-sensitive} we will build a $w$-ary z-fast trie (that is, a z-fast trie of arity $w$ endowed with a fusion tree on each explicit node supporting constant-time predecessor queries on the outgoing edges of the node) on the base-$w$ representations of the integers $x_1, \dots, x_s$. 
Classic \emph{fat binary search} on this trie yields predecessor queries in time $O(\log\log_wn)$. In order to speed this up to  $O(\log\log_w\Delta)$, we generalize the exponential search of \cite{BBV12,ehrhardt2017delta} (designed on binary z-fast tries) to $w$-ary z-fast tries.
The exponential search phase will precede classic fat binary search and will find, in $k\in O(\log\log_w\Delta)$ steps, a trie node at height (distance from the leaves) $2^k \in O(\log_w\Delta)$, hence fat binary search from that node will find the predecessor of $q$ in $O(\log 2^k) \subseteq O(\log\log_w\Delta)$ time.

In the following, for convenience we switch to base $w$ when dealing with integers (this will avoid an excessive use of ceiling operations and will simplify notation), where $w$ is the machine word size in bits. More precisely, we work with integers from $[0,w^H)$ ($H$ digits in base $w$), for the smallest integer $H$ being a power of two and satisfying $H \ge \log_w 2^w$. Observe that any $q\in [0,w^H)$ still fits in $O(1)$ memory words, so this assumption is without loss of generality. We will moreover use string notation on the integers from $[0,w^H)$ (as well as standard integer notation).
Given $q \in [0,w^H)$ and $i\in [1,H]$,  $q[i]$ denotes the $i$-th most significant digit of $q$ written in base $w$. In other words, $q$ is treated as a string of length $H$ over alphabet $[0,w)$. 

Given $X = \{x_1, \dots, x_s\} \subseteq [1,w^H)$, we denote by $\mathit{pref}(X)$ the set containing all prefixes of strings in $X$.

Consider the $w$-ary trie containing $x_1, \dots, x_s$ (seen as strings of length $H$ over alphabet $[0,w)$). Letting $\lambda(u) \in \mathit{pref}(X)$ denote the label from the root to trie node $u$, $u$ is \emph{explicit} if and only if either $|\lambda(u)|=H$ or if $\lambda(u)\cdot a$ and $\lambda(u)\cdot b$ belong to $\mathit{pref}(X)$ for two distinct $a, b \in [0,w)$. 
We path-compress this trie and denote by $Z$ the resulting trie; nodes of $Z$ correspond to the explicit nodes of the original trie and edges of $Z$ are labeled with sequences of digits from $[0,w)$.

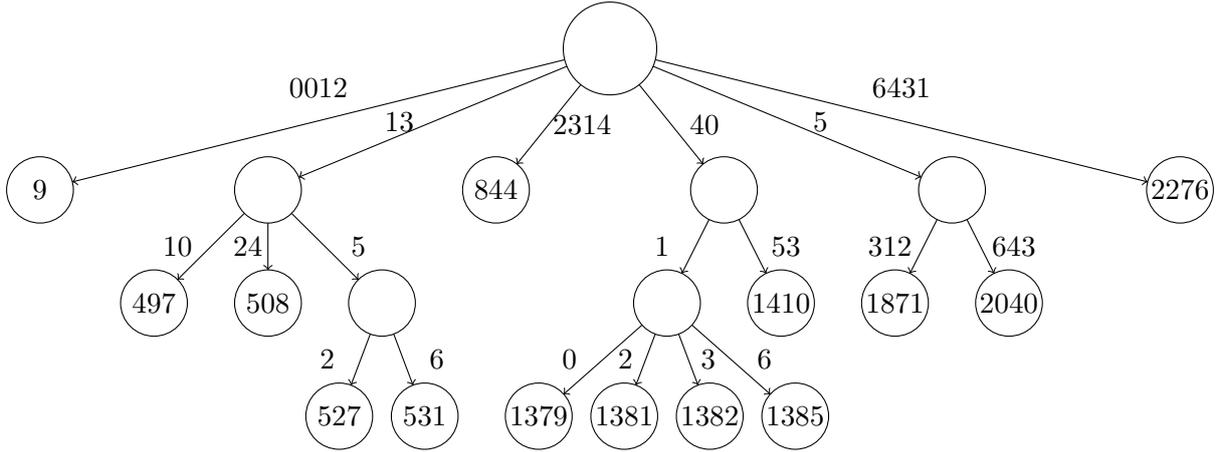
\begin{figure}[t]
\center
\begin{tikzpicture}[scale=0.75,minimum size=25pt,inner sep=0pt, outer sep=0pt]
    \node[circle, draw, minimum size=35pt] (root) at (0,0.5) {};
    \node[circle, draw] (0012) at (-10,-2) {$9$};
    \node[circle, draw] (13) at (-6,-2) {$ $};
    \node[circle, draw] (2314) at (-2,-2) {$844$};
    \node[circle, draw] (40) at (2,-2) {$ $};
    \node[circle, draw] (5) at (6,-2) {$ $};
    \node[circle, draw] (6431) at (10,-2) {$2276$};

    \node[circle, draw] (1310) at (-8,-4) {$497$};
    \node[circle, draw] (1324) at (-6,-4) {$508$};
    \node[circle, draw] (135) at (-4,-4) {$ $};
    \node[circle, draw] (5312) at (5,-4) {$1871$};
    \node[circle, draw] (5643) at (7,-4) {$2040$};
    \node[circle, draw] (1352) at (-4.75,-6) {$527$};
    \node[circle, draw] (1356) at (-3.25,-6) {$531$};

    \node[circle, draw] (401) at (1,-4) {$ $};
    \node[circle, draw] (4053) at (3,-4) {$1410$};
    \node[circle, draw] (4010) at (-1.25,-6) {$1379$};
    \node[circle, draw] (4012) at (0.25,-6) {$1381$};
    \node[circle, draw] (4013) at (1.75,-6) {$1382$};
    \node[circle, draw] (4016) at (3.25,-6) {$1385$};

    \draw[->] (root) -- node[above]{$0012$} (0012);
    \draw[->] (root) -- node[left]{$13$} (13);
    \draw[->] (root) -- node[right]{$2314$} (2314);
    \draw[->] (root) -- node[right]{$40$} (40);
    \draw[->] (root) -- node[right]{$5$} (5);
    \draw[->] (root) -- node[above]{$6431$} (6431);

    \draw[->] (13) -- node[left]{$10$} (1310);
    \draw[->] (13) -- node[left=-5pt]{$24$} (1324);
    \draw[->] (13) -- node[right]{$5$} (135);

    \draw[->] (135) -- node[left]{$2$} (1352);
    \draw[->] (135) -- node[right]{$6$} (1356);

    \draw[->] (40) -- node[left]{$1$} (401);
    \draw[->] (40) -- node[right]{$53$} (4053);
    \draw[->] (401) -- node[left]{$0$} (4010);
    \draw[->] (401) -- node[left=-5pt]{$2$} (4012);
    \draw[->] (401) -- node[right=-5pt]{$3$} (4013);
    \draw[->] (401) -- node[right]{$6$} (4016);

    \draw[->] (5) -- node[left]{$312$} (5312);
    \draw[->] (5) -- node[right]{$643$} (5643);
\end{tikzpicture}\caption{Example of the compacted trie $Z$ with word size $w = 7$ and $H = 4\ge\log_72^7\approx2.48$ for the set $X=\{9,497,508,527,531,844,1379,1381,1382,1385,1410,1871,2040,2276\}$. In the leaf nodes there are the values of $X$. The concatenation of the labels from the root to a leaf node is the representation in base $w$ of the value in that node.}
\label{fig:2}
\end{figure}

In the following, for brevity, when saying that edge $e$ (of trie $Z$) is labeled $x_i[l,r]$, we also mean that the string read from the root of $Z$ to the end of $e$ is $x_i[1,r]$. 

We identify one \emph{canonical position} $f(l,r)$ inside each edge labeled $x_i[l,r]$ of the path-compressed trie $Z$, as follows:

\begin{definition}
    The \emph{2-fattest number} $f(l,r)$ in a given interval $[l,r]$ ($1 \le l\le r \le w$) is the number  $f(l,r) \in [l,r]$ whose binary representation has the largest number of trailing zeros among all numbers in $[l,r]$.
\end{definition}

It is not hard to see that the 2-fattest number in any given interval is unique, because for any pair of integers $i<j$ whose binary representation has $k$ trailing zeros, there always exists a third integer $z$ with $i<z<j$ whose binary representation has $k+1$ zeros. 
Intuitively (as we will see later), this particular choice of canonical position has the additional advantage of enabling a binary search strategy on the trie's height.

Based on the above definition, we define the \emph{handle} of an edge as follows.

\begin{definition}[handle]
    Let $e$ be a trie edge labeled with $x_i[l,r]$. The \emph{handle} $h_e$ of $e$ is the string $h_e = x_i[1,f(l,r)]$.
\end{definition}

The next ingredient towards defining our data structure is a map storing information about children of nodes in $Z$. 

\begin{definition}
    Let $e$ be a trie edge labeled with $x_i[l,r]$. The (partial) function $\mathcal C_e : [0,w) \rightarrow [0,w)^*$ stores, for each $c \in [0,w)$ such that $x_i[1, r]\cdot c \in \mathit{pref}(X)$, the string $\mathcal C_e(c) = x_{i'}[1, r']$ prefixed by $x_i[1, r]\cdot c$ such that there exists a trie edge labeled $x_{i'}[r+1,r']$ for some $i',r'$. 
    We denote with $\mathcal C_\epsilon$ the additional special case where $x_i[1, r=0]=\epsilon$ is the empty string. 
\end{definition}

In other words, the domain of $\mathcal C_e$ corresponds to the characters that can extend edge $e$ in the trie. 
For each such character $c$, the map $\mathcal C_e$ stores the label $\mathcal C_e(c)$ of the path starting in the trie's root and ending at the end of edge $e'$ that is a child of edge $e$ and whose label starts with $c$. 
The particular case $\mathcal C_\epsilon$ corresponds to the trie's root: $\mathcal C_\epsilon(c)$ is the label of the edge exiting the root and whose label starts with character $c$.

Finally, our structure will associate the following additional information with the trie's edges.

\begin{definition}\label{def:precomputed}
    Let $x_i[l,r]$ be the label of an edge of $Z$. Then:
    \begin{itemize}
        \item $x^-_{i,r}$ is the strict predecessor in $X \cup \{0\}$ of the smallest $y\in X$ being prefixed by $x_i[1, r]$; 
        \item $x^+_{i,r} : [0,w)\cup\{\epsilon\} \rightarrow X$ is the (partial) function associating with every $c\in [0,w) \cup \{\epsilon\}$ the largest number $x^+_{i,r}(c) \in X$ that is prefixed by $x_i[1,r]\cdot c$ (if no such number exists, then the map is not defined on $c$).
    \end{itemize}
\end{definition}

The $w$-ary $z$-fast trie of $x_1, \dots, x_s$ is a function $\mathcal Z$ defined as follows. 

\begin{definition}[$w$-ary $z$-fast trie]\label{def:z-fast}
The $w$-ary $z$-fast trie of $x_1, \dots, x_s$ is the (partial) function $\mathcal Z$ defined as follows. For every edge $e$ of $Z$ labeled with $x_i[l,r]$ (for some $i,l,r$), we define $\mathcal Z(h_e) = (r, \mathcal C_e, x^-_{i,r}, x^+_{i,r})$, where $h_e = x_i[1,f(l,r)]$ is the handle of $e$.
\end{definition}

Note that $x_i[1, r] = x^+_{i,r}(\epsilon)[1,r]$, so below we take for granted that such a string can be retrieved in constant time from $\mathcal Z(h_e)$.

We implement $\mathcal Z$ using perfect hashing. 
Observe that membership in the domain of $\mathcal Z$ can be determined in $O(1)$ time since $\mathcal Z(h_e)$ can be used to obtain $x_i[1, r]$, which is an extension of $h_e$.
Functions $\mathcal C_e$ and $x^+_{i,r}$ are returned in $\mathcal Z(h_e)$ as pointers (so retrieving $\mathcal Z(h_e)$ takes constant time), and are implemented using a fusion tree with constant-time predecessor queries over their domains.

Function $\mathcal Z$ is defined on $O(s)$ distinct arguments $h_e$ (one for each trie's edge), hence integers $r$ and $x^-_{i,r}$ returned in $\mathcal Z(h_e) = (r, \mathcal C_e, x^-_{i,r}, x^+_{i,r})$ use overall $O(s)$ words of space. 
Similarly, $\mathcal C_e$ is defined on $n_e$ distinct arguments, where $n_e$ is the number of one-character right-extensions of $x_i[1,r]$ ($x_i[l,r]$ being the label of edge $e$) yielding an element of $\mathit{pref}(X)$; for each such argument $c$, $\mathcal C_e(c) \in [0,w)^{\le H}$ and therefore it can be packed in $O(1)$ machine words. Being also those right-extensions in bijection with the trie's edges, we conclude that all those functions $\mathcal C_e$ use $O(s)$ space overall. The same reasoning holds for functions $x^+_{i,j}$. Overall, the $w$-ary $z$-fast trie for $X$ uses $O(s)$ words of space. 

The structure of Definition \ref{def:z-fast} can be used to solve predecessor queries in $O(\log\log_w u)$ time ($u=w^H$). In the following subsection we prove a more general result that will allow us achieving distance-sensitive time $O(\log\log_w \Delta)$.

\subsection{Extending to obtain distance-sensitive times}

In order to support distance-sensitive predecessor queries, we augment the $z$-fast trie of Definition \ref{def:z-fast} with one additional data structure.

\begin{definition}\label{def:qk}
Given $q \in [0,w^H)$ and $k \in [0,\log H]$, we denote by $q_k$ the string $q[1,H]$ devoid of its suffix of length $2^k-1$, that is, $q_k = q[1,H-2^k+1]$. 
\end{definition}

Note that, in Definition \ref{def:qk}, it holds $k\in O(\log H) = O(\log\log_w u)$, where $u = w^H$ is the universe size, and $q_k$ is a prefix of $q$ that gets shorter as $k$ increases.
In particular, $q_0 = q$ and $q_{\log H} = q[1]$.
As said above, $q_k$ will be treated equivalently as an integer formed by the $H-2^k+1$ most significant digits of $q$ in base $w$ (this will allow us to perform integer subtractions on it).  

\begin{definition}\label{def:exponential map}
Let $\mathcal H : [0,w)^* \rightarrow [0,w)^*$ be the (partial) function defined as follows. For each $k \in [0,\log H]$ and each $q \in X$, $\mathcal H(q_k)=|h_e|$ is the length of the longest handle $h_e$ that is a prefix of $q_k$. If no such edge $e$ exists, then $\mathcal H(q_k)=0$.
\end{definition}

We store $\mathcal H$ with perfect hashing. Observe that $k$ in Definition \ref{def:exponential map} satisfies $k\in O(\log H) = O(\log\log_w 2^w) \subseteq O(\log w)$, therefore $\mathcal H$ is defined on $O(s\log w)$ distinct arguments. Similarly, $\mathcal H(q_k) = |h_e| \le H \in O(\log_w2^w)$, therefore it can be packed in $O(\log w)$ bits. We conclude that $\mathcal H$ can be stored in $O(s(\log w)^2) \subseteq O(sw)$ bits, or $O(s)$ words.

\begin{lemma}\label{lem:using map H}
    Functions $\mathcal Z$ and $\mathcal H$ take $O(s)$ words of space and allow determining in $O(1)$ time whether $p \in \mathit{pref}(X)$, for any $p \in [0,w)^{H-2^k+1}$ and $k \in [0,\log H]$. 
\end{lemma}
\begin{proof}
Let $p \in [0,w)^{H-2^k+1}$, with $k \in [0,\log H]$. We consider two cases.

(1) $\mathcal H(p)=0$. If $p\in \mathit{pref}(X)$, then $p$ must prefix the label $x_i[1,r]$ of one of the edges exiting from the trie's root, for some $i,r$. To check that this is the case, we get the label $x_i[1,r] = \mathcal C_\epsilon(p[1])$ of that (candidate) edge. 
If $|p|>r$ then $p\notin \mathit{pref}(X)$. Otherwise, $p\in \mathit{pref}(X)$ if and only if $x_i[1,|p|] = p$. 

(2) $\mathcal H(p)>0$. We know the length $|h_e| = \mathcal H(p)$ of the handle of some edge $e$ labeled with $x_i[l,r]$ maximizing $r$ and such that $h_e$ prefixes $p$ --- up to collisions in the perfect hash representing $\mathcal H$, which we now show how to detect. 
If $|h_e|>|p|$ then this is indeed a hash collision and $p\notin \mathit{pref}(X)$. Otherwise, we compute $h_e = p[1,|h_e|]$ in constant time and get $\mathcal Z(h_e) = (r, \mathcal C_e, x^-_{i,r}, x^+_{i,r})$, also in constant time. If $|p| \le r$, then we check in constant time that $p = x_i[1,|p|]$ (recall that $\mathcal Z(h_e)$ allows us retrieving $x_i[1,r]$); this test succeeds if and only if $p\in \mathit{pref}(X)$. Otherwise ($|p| > r$), let $c = p[r+1]$. We compute in constant time $\mathcal C_e(c) = x_{i'}[1, r']$, which at this point must be an extension of $p$ if $p$ does indeed belong to $\mathit{pref}(X)$. We check in constant time that $p = x_{i'}[1, |p|]$; this test succeeds if and only if $p\in \mathit{pref}(X)$.
\end{proof}

Running Lemma \ref{lem:using map H} for $k=0,1,2, \dots$ corresponds to running an exponential search for the longest prefix of $q$ matching an element in $\mathit{pref}(X)$, obtaining the following result. Searching also for $q_k-1$ will be needed later to solve predecessor queries.

\begin{corollary}[exponential search]\label{cor:expsearch}
    Let $X = \{x_1, \dots, x_s\} \subseteq [1,w^H)$.
    We can build a data structure on $X$ taking $O(s)$ words of space that, given any $q \in [1,w^H)$, returns the pair $(k,q')$ defined as follows. Letting $P_k=\{q_k,q_k-1\} \cap (\mathit{pref}(X)\cup \{0\})$, $k$ is the smallest 
    integer from $[0,\log H]$ such that $P_k \neq \emptyset$, and $q' = \max P_k$.
    If no such integer $k$ exists, it returns $(-1,-1)$.
    This query is answered in $O(1+k)$ time.
\end{corollary}

In Corollary \ref{cor:expsearch}, note that we allow $q_k$ (or $q_{k-1}$) to prefix also 0. This means that we stop looking for the smallest $k$ even if $q_k$ (or $q_{k-1}$) prefixes 0 but no other element in $X$. Additionally, observe that it is implicit that if $q_k=0$, then $q_k-1 = -1$ does not belong to $\mathit{pref}(X)\cup\{0\}$ (this avoids treating that particular case separately).  Observe that Corollary \ref{cor:expsearch} finds $k$ by linear search on $k$ (that is, $k=0,1,2,\dots$). This is enough to reach the desired running times for predecessor queries, as we will prove that $k \in O(\log\log_w\Delta)$.

\begin{lemma}[$w$-ary z-fast trie with hints]\label{lem:z-fast}
    Let $X = \{x_1, \dots, x_s\} \subseteq [1,w^H)$. 
    There exists a data structure on $X$ taking $O(s)$ words of space supporting the following query.
    Given $q \in [1,w^H)$ and the smallest $k \in [0,\log H]$ (we call $k$ the \emph{hint}) such that $q_k\in \mathit{pref}(X)$, find the predecessor $q^-$ of $q$ in $X\cup\{0\}$. This query is solved in $O(1+k)$ time. 
\end{lemma}
\begin{proof}
    We know that there exists an integer $q'\in X$ prefixed by $q_k = q[1,H-2^k+1]$. If $k=0$ then $q = q_k  \in X$, hence we just return $q$ and we are done, so in the following we assume $k>0$.

    We show how to use classic fat binary search on the z-fast trie to find the longest handle of length in $[H-2^k+2, H-1]$ that is a prefix of $q$ (note that we can exclude length $H$ since $k>0$) in the claimed running time. A first observation is that (since $H$ is excluded from that interval) the binary representation of any number in $[H-2^k+2, H-1]$ has no more than $k$ trailing zeros. This enables finding the answer by $O(1+k)$ steps of fat binary search as follows. We check, using map $\mathcal Z$, whether $q[1,f(H-2^k+2, H-1)]$ is a handle. If the answer is positive, then we recurse in interval $[f(H-2^k+2, H-1)+1,H-1]$. Otherwise, we recurse in interval $[H-2^k+2, f(H-2^k+2, H-1)-1]$. Since each recursive step removes the 2-fattest number from the interval, it also decreases the number of trailing zeros of the 2-fattest number in the interval by one unit. We conclude that this procedure finds the longest handle $y$ prefixing $q$ in $O(1+k)$ time.

    Let $\ell = lcp(q,X) \in [0,w)^{<H}$ denote the longest prefix of $q$ belonging to $\mathit{pref}(X)$.
    From the handle $y$ found in the previous step, using map $\mathcal Z$ one can get in constant time (1) the label $x_i[1, r]$ ending at the end of trie edge $e$ labeled with $x_i[l,r]$ (for some $i,l,r$) such that $l \le |\ell| \le r$ and $\ell = x_i[1, |\ell|]$, as well as (2) dictionary $\mathcal C_e$, (3) value $x^-_{i,r}$ and (4) function $x^+_{i,r}$.
    In other words, by reading $\ell$ from the trie's root we end up inside edge $e$. We distinguish three cases. 
    
    (1) If $x_i[1, r]<q[1, r]$, then $q^-$ is the largest number in $X$ being prefixed by $x_i[1, r]$, that is, $x^+_{i,r}(\epsilon)$.
    
    (2) If $x_i[1, r]>q[1, r]$, then $q^- = x_{i,r}^-$.
    
    (3) Finally, it could be $x_i[1, r]=q[1,r]$, meaning that $\ell$ matches until the end of edge $e$. In that case, let $c = q[r+1]$ be the mismatching character (note that $c$ is well-defined since $r = |\ell| < H$). 
    We find the predecessor $c^-$ of $c$ among the characters in the domain of $\mathcal C_e$, that is, characters extending edge $e = x_i[l,r]$ in the trie. 
    This operation takes constant time, since $\mathcal C_e$ is represented with a fusion tree.
    If such $c^-$ exists, then $q^- = x^+_{i,r}(c^-)$. Otherwise, $c$ is smaller than all those characters and $q^-$ is again  $x_{i,r}^-$ as in case (2).
 \end{proof}

\subsection{Putting everything together}

Given $q \in [1,w^H)$, we apply Corollary \ref{cor:expsearch} and find, in time $O(1+k)$, the smallest value $k$ such that either $q_k$ or $q_k-1$ prefixes some integer in $X\cup\{0\}$, while (if $k>0$) both $q_{k-1}$ and $q_{k-1}-1$ do not. We prove that such a value $k$ is distance-sensitive.

\begin{lemma}\label{lem:order of k}
Let $k$ be the value identified by Corollary \ref{cor:expsearch}. 
Then, 
$k\in O(\log\log_w\Delta)$.     
\end{lemma}
\begin{proof}
If $k=0$ then the claim is trivially true, so we assume $k>0$.
Observe that by the definition of $k$ it must be $P_{k-1}=\{q_{k-1},q_{k-1}-1\} \cap (\mathit{pref}(X)\cup \{0\}) = \emptyset$, hence in particular $q_{k-1} > 0$ and $q_{k-1}-1 \neq 0$ hold. This allows us to conclude that $q_{k-1}-1 > 0$.

Integers prefixed by either $q_{k-1}$ or $q_{k-1}-1$ form a contiguous integer range $R = R_1\cup R_2$ where $R_1 = [(q_{k-1}-1)\cdot w^{2^{k-1}-1},q_{k-1}\cdot w^{2^{k-1}-1})$ and $R_2 = [q_{k-1}\cdot w^{2^{k-1}-1},(q_{k-1}+1)\cdot w^{2^{k-1}-1})$ are disjoint: $R_1\cap R_2 = \emptyset$. 
Since no element in $X\cup\{0\}$ is prefixed by neither $q_{k-1}$ nor $q_{k-1}-1$, we have that (A) $R \cap (X\cup\{0\}) = \emptyset$. 
Moreover, since $q$ is prefixed by $q_{k-1}$, then (B) $q \in R_2$.  
Observations (A) and (B) imply that $q^-$ comes strictly before the left bound of $R_1$, that is, $q^- < (q_{k-1}-1)\cdot w^{2^{k-1}-1}$.
Then, $\Delta = q - q^- > |R_1| = w^{2^{k-1}-1}$. This implies $k < 1 + \log(1+\log_w\Delta) \in O(\log\log_w \Delta)$.
\end{proof}

Let $k$ be the value identified by Corollary \ref{cor:expsearch}. We distinguish two cases. 

(1) If $q_k-1$ prefixes some integer in $X\cup\{0\}$ but $q_k$ does not (in particular, $q_k-1 \ge 0$), then $q^-$ is the largest integer in $X\cup \{0\}$ being prefixed by $q_k-1$. We find such a value as follows.
We check if $q_k-1 \in \mathit{pref}(X)$ using 
Lemma \ref{lem:using map H}.
If $q_k-1$ does not prefix any integer in $X$, then $q^-=0$ and we are done. Otherwise, 
using maps $\mathcal H$ and $\mathcal Z$ we identify in constant time the edge $e$ labeled with $x_i[l,r]$ such that $x_i[1,r]$ is prefixed by $q_k-1$ but $x_i[1,l-1]$ is not, as well as its associated map $x^+_{i,r}$. Then, $q^- = x^+_{i,r}(\epsilon)$.

(2) Otherwise, $q_k$ prefixes some integer in $X \cup \{0\}$.
If $q_k \notin \mathit{pref}(X)$ (a property we can check using Lemma \ref{lem:using map H}), then $q^- = 0$. Otherwise, we can apply Lemma \ref{lem:z-fast} on query $q$ with hint $k$ and find $q^-$ in   $O(1+k)$ time. 

Overall, finding $q^-$ takes $O(1+k)$ time. By Lemma \ref{lem:order of k}, this is $O(1+k) = O(\log\log_w\Delta)$.
This proves Theorem \ref{thm:distance-sensitive}.

\section{Figure for Lemma~\ref{lem:leaves2}}
\label{app:fig_for_leaves2}

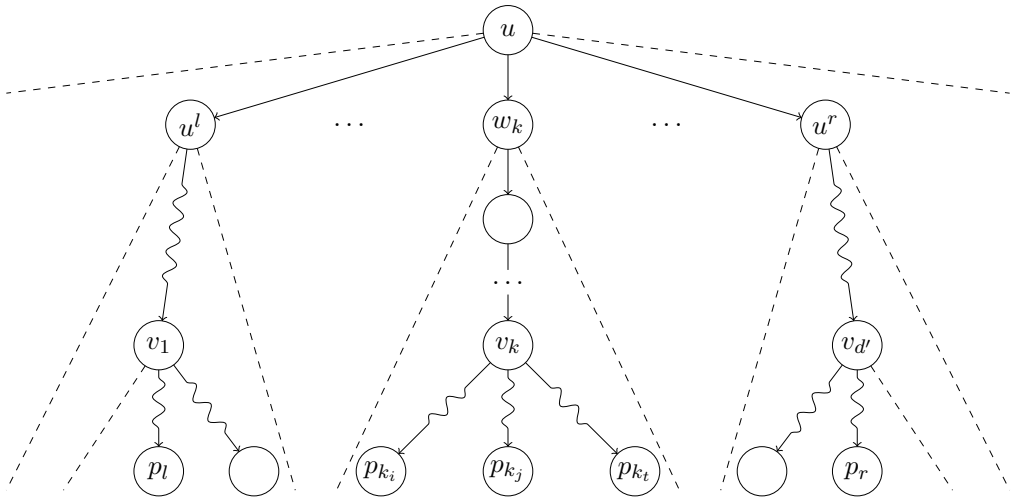
\begin{figure}[h!]
\center
\resizebox{.83\textwidth}{!}
{\begin{tikzpicture}[scale=0.9,decoration=snake,
   line around/.style={decoration={pre length=#1,post length=#1}},
   minimum size=20pt,inner sep=0pt, outer sep=0pt]
    \node[circle, draw] (u) at (0,0) {$u$};
    \node (dotsl) at (-2.5,-1.5) {$\dots$};
    \node (dotsr) at (2.5,-1.5) {$\dots$};

    \node[circle, draw] (ul) at (-5,-1.5) {$u^l$};
    \node[circle, draw] (v1) at (-5.5,-5) {$v_1$};
    \node[circle, draw] (v1_lth_leaf) at (-5.5,-7) {$p_l$};
    \node[circle, draw] (v1_right) at (-4,-7) {$ $};
    \node (ul_left_limit) at (-7,-7) {};

    \node[circle, draw] (wk) at (0,-1.5) {$w_k$};
    \node[circle, draw] (vk_1) at (0,-3) {};
    \node[minimum size=10pt] (vk_dots) at (0,-4) {$\dots$};
    \node[circle, draw] (vk) at (0,-5) {$v_k$};
    \node[circle, draw] (vk_left) at (-2,-7) {$p_{k_i}$};
    \node[circle, draw] (vk_middle) at (0,-7) {$p_{k_j}$};
    \node[circle, draw] (vk_right) at (2,-7) {$p_{k_t}$};

    \node[circle, draw] (ur) at (5,-1.5) {$u^r$};
    \node[circle, draw] (vd') at (5.5,-5) {$v_{d'}$};
    \node[circle, draw] (vd'_rth_leaf) at (5.5,-7) {$p_r$};
    \node[circle, draw] (vd'_left) at (4,-7) {$ $};

    \draw[->] (u) --  (ul);
    \draw[->,decorate,line around=0.5cm]  (ul) -- ++(v1);
    \draw[->,decorate,line around=0.1cm]  (v1) -- ++(v1_lth_leaf);
    \draw[->,decorate,line around=0.3cm]  (v1) -- ++(v1_right);
    \draw[-, dashed] (ul) -- (-7.9,-7.3);
    \draw[-, dashed] (ul) -- (-3.35,-7.3);
    \draw[-, dashed] (v1) -- (-7,-7.3);

    \draw[->] (u) --  (wk);
    \draw[->]  (wk) -- (vk_1);
    \draw[-]  (vk_1) -- (vk_dots);
    \draw[->]  (vk_dots) -- (vk);
    \draw[->,decorate,line around=0.5cm]  (vk) -- ++(vk_left);
    \draw[->,decorate,line around=0.1cm]  (vk) -- ++(vk_middle);
    \draw[->,decorate,line around=0.5cm]  (vk) -- ++(vk_right);
    \draw[-, dashed] (wk) -- (-2.7,-7.3);
    \draw[-, dashed] (wk) -- (2.7,-7.3);
    \draw[-, dashed] (u) -- (-7.9,-1);
    \draw[-, dashed] (u) -- (7.9,-1);

    \draw[->] (u) --  (ur);
    \draw[->,decorate,line around=0.5cm]  (ur) -- ++(vd');
    \draw[->,decorate,line around=0.1cm]  (vd') -- ++(vd'_rth_leaf);
    \draw[->,decorate,line around=0.3cm]  (vd') -- ++(vd'_left);
    \draw[-, dashed] (ur) -- (7.9,-7.3);
    \draw[-, dashed] (ur) -- (3.35,-7.3);
    \draw[-, dashed] (vd') -- (7,-7.3);
\end{tikzpicture}}
\caption{Subtrees of the trie $T$ of Lemma \ref{lem:subintervals2} with roots $v_1,\dots, v_{d'}$ at depth $\log_d \frac{n}{x_r + x_{r + 1} - x_l  - x_{l  + 1}}$, containing the $l$th through $r$th leaves $p_l, \dots, p_r$ of $T$ (as explained in the proof of Lemma \ref{lem:leaves2}). }\label{fig:trie_T}
\end{figure}

\section{Proof of Lemma~\ref{lem:main2}}
\label{app:main2}

\begin{proof}
For each subinterval $[y_i, z_i)$, we store the index $j'$ of the smallest $x_{j'} \ge y_i$ and the index $j''$ of the largest $x_{j''} < z_i$.  This way, if $[x_j, x_{j + 1})$ is not completely contained, or it is the only range contained, in $[y_i, z_i)$ then when given $i$ and $q$ we can find $j$ with $O(1)$ integer comparisons and no accesses to $T$.

Otherwise, let $l<r$ be such that $[x_l,x_{l+1}),\ldots,[x_r,x_{r+1})$ are the ranges fully contained in $[y_i,z_i)$. We store, associated with every interval $[y_i,z_i)$, the three nodes $v_1, v_{d'}$, and $u$ of $T$ identified in Lemma~\ref{lem:leaves2} for the values $l$ and $r$. 
We additionally store the two children $v'$ and $v''$ of $u$ that are ancestors of $v_1$ and $v_{d'}$, respectively (those were called $u^l$ and $u^r$ in the proof of Lemma~\ref{lem:leaves2}). 
Further, for each node $u$ of $T$ with at least two children $u_1, \dots, u_{d''}$ we build a fusion tree; for every such child $u_p$,
the fusion tree stores the integer $x_a$ such that the leftmost leaf below $u_p$ is associated with interval $[x_a,x_{a+1})$. 
Because $T$ has $s$ leaves and the fusion trees are built on the nodes with at least two children, those fusion trees use $O(s)$ space in total.

Given an integer $q \in [y_i,z_i)$ and the node $u$ associated with $[y_i,z_i)$, the fusion tree of $u$ returns in $O(\log_w d)$ time the child $u_p$ of $u$ such that the $j$th leaf (i.e., the one such that $q \in [x_j,x_j+1)$) is below $u_p$. If $u_p$ is equal to $v'$ or $v''$, we replace it with $v_1$ or $v_{d'}$, respectively. Otherwise, 
we replace $u_p$ with the deepest descendant of $u_p$ being an ancestor of all the leaves below $u_p$; such substitution can be performed in constant time by explicitly associating such a node to each $u_p$ in $T$ (again, using $O(s)$ space overall).
At this point, observe that $u_p$ is precisely the node among $v_1, \dots, v_{d'}$ identified by Lemma~\ref{lem:leaves2} such that (i) the $j$th leaf is below $u_p$ and (ii) the depth of $u_p$ in $T$ is at least $\log_d \frac{n}{x_r + x_{r + 1} - x_l - x_{l + 1}}$.

Since $u_p$ has depth at least 
\begin{eqnarray*}
\lefteqn{\log_d \frac{n}{x_r + x_{r + 1} - x_l - x_{l + 1}}} \\
& = & \log_d \frac{n}{(x_r + x_{r + 1}) / 2 - (x_l + x_{l + 1}) / 2} - \log_d2 \\
& \geq & \log_d \frac{n}{z_i - y_i} - 1
\end{eqnarray*}
and $x_j$ is at depth
\[\left\lfloor \log_d \frac{n}{x_{j + 1} - x_j} \right\rfloor\, + 2,\]
we can find the $j$th leaf of $T$ by starting from $u_p$, traversing at most 
\[\log_d \frac{n}{x_{j + 1} - x_j} - \log_d \frac{n}{z_i - y_i} + 3
~~=~~  \log_d \frac{z_i - y_i}{x_{j + 1} - x_j} + O(1)\,\]
edges of $T$.
\end{proof}

\section{Proof of Lemma~\ref{lem:alpha bidirectional}}
\label{app:alpha bidirectional}

\begin{proof}

First, we build the same string attractor $\Gamma$ of \cite{KP18} --- of size $O(b)$ --- on the input parse: for every phrase $S[i\dd j]$, we insert $i$ and $j$ in $\Gamma$. Then, we additionally insert in $\Gamma$ integers $1$, $n$, and $k\cdot \lceil n/b \rceil$ in $\Gamma$ for every $k= 1, 2, \dots b-1$. As a result, $\Gamma$ is a string attractor for $S$ of size $O(b)$ with the property that pairs of consecutive string attractor positions are within $O(n/b)$ positions from each other. Let $\Gamma = \{x_1 < \cdots < x_p\}$ be the string attractor obtained in this way.

We use a technique similar to the concentric exponential parsing \cite[Thm 3.12]{KP18} and build a bidirectional parse for $S$ as follows. Repeat the following construction for every pair of consecutive positions $x_i < x_{i+1}$ in $\Gamma$. Let $m = (x_i + x_{i+1})/2$. We make $S[x_i]$, $S[x_{i+1}]$, and $S[m]$ explicit phrases. Then, we parse $S[x_i+1 \dd m-1]$ left-to-right creating one (non-explicit) phrase of length $\alpha^1$, followed by one (non-explicit) phrase of length $\alpha^2$, and so on as long as the current phrase we are building is completely contained inside $S[x_i+1 \dd m-1]$. 
    We repeat a symmetric procedure parsing right-to-left the substring $S[m+1 \dd x_{i+1}]$. 
    These two procedures may leave a region spanning $S[m]$ not being covered by any phrase; we make that region a single phrase (observe that the length of this phrase is not necessarily a power of $\alpha$). 
    The above procedures create in total $O(b \log_\alpha(n/b))$ phrases. 

    Phrases' sources are chosen as follows. Observe that, by construction, each phrase in the parse we just built is strictly contained in a phrase of the original parse. Given a phrase $S[i\dd j]$ in the parse we just built, we follow its source $S[i'\dd j'] = S[i\dd j]$ in the original parse. If $S[i'\dd j']$ spans an element of $\Gamma$ then we stop; 
    otherwise, we recursively repeat following sources until this becomes true (by definition of bidirectional parse, this procedure cannot loop forever). 
    Let $S[i''\dd j''] = S[i\dd j]$ be the substring (containing an element of $\Gamma$) found in this way. We assign $S[i''\dd j'']$ as source of $S[i\dd j]$. Observe that this source assignment is such that: (1) the height $h'_q$ of any $S[q]$ in the new parse is no larger than that of the original parse: $h'_q \le h_q$, and (2) the resulting parse is indeed a valid bidirectional parse (i.e., no loops of referencing chains are introduced). Both properties follow from the fact that we possibly just ``short-cut'' referencing chains of the original parse, making them (possibly) shorter.
    
    We now show that the source of every non-explicit phrase of length in $[\alpha^t,\alpha^{t+1})$ overlaps only phrases of length at most $\alpha^{t+1}$, implying that the parse is $\alpha$-contracting.
    To see this, let $S[i \dd i+d-1]$ be such a source ($d\in [\alpha^t,\alpha^{t+1})$), and let $j \in [i \dd i+d-1] \cap \Gamma$ be any attractor element crossing it. Consider the source's suffix $S[j \dd i+d-1]$ (a symmetric argument holds for the prefix $S[i \dd j-1]$), and let $\lambda = i+d-j \in [\alpha^t,\alpha^{t+1})$ be its length. 
    By construction of our parse, the leftmost phrase of length at least $\alpha^{t+2}$ in $S[j \dd n]$ cannot start before $S[j+k]$, where $k = \sum_{z=0}^{t+1} \alpha^z > \alpha^{t+1}$. We conclude that $S[j \dd i+d-1]$ ends before $S[j+k]$, hence it can only overlap phrases of length at most $\alpha^{t+1}$. Our claim follows.
\end{proof}

\end{document}